\documentclass[11pt]{amsart}
\usepackage{amsfonts,amssymb, graphicx,a4}
\usepackage{color}\definecolor{Red}{cmyk}{0,1,1,0}

\numberwithin{equation}{section}

\usepackage{bbm}
  \newtheorem*{theorem*}        {Theorem}
	\newtheorem*{conjecture*}   {Conjecture}
  \newtheorem{theorem}           {Theorem}
  
  \newtheorem*{lemma*}          {Lemma}
    \newtheorem*{claim*}          {Claim}
  \newtheorem{definition}         {Definition}
  \newtheorem{corollary}          {Corollary}
  \newtheorem{proposition}      {Proposition}

\begin{document}


\voffset=-1.5truecm\hsize=16.5truecm    \vsize=24.truecm
\baselineskip=14pt plus0.1pt minus0.1pt \parindent=12pt
\lineskip=4pt\lineskiplimit=0.1pt      \parskip=0.1pt plus1pt

\def\ds{\displaystyle}\def\st{\scriptstyle}\def\sst{\scriptscriptstyle}

\global\newcount\numsec\global\newcount\numfor
\gdef\profonditastruttura{\dp\strutbox}
\def\senondefinito#1{\expandafter\ifx\csname#1\endcsname\relax}
\def\SIA #1,#2,#3 {\senondefinito{#1#2}
\expandafter\xdef\csname #1#2\endcsname{#3} \else
\write16{???? il simbolo #2 e' gia' stato definito !!!!} \fi}
\def\etichetta(#1){(\veroparagrafo.\veraformula)
\SIA e,#1,(\veroparagrafo.\veraformula)
 \global\advance\numfor by 1
 \write16{ EQ \equ(#1) ha simbolo #1 }}
\def\etichettaa(#1){(A\veroparagrafo.\veraformula)
 \SIA e,#1,(A\veroparagrafo.\veraformula)
 \global\advance\numfor by 1\write16{ EQ \equ(#1) ha simbolo #1 }}
\def\BOZZA{\def\alato(##1){
 {\vtop to \profonditastruttura{\baselineskip
 \profonditastruttura\vss
 \rlap{\kern-\hsize\kern-1.2truecm{$\scriptstyle##1$}}}}}}
\def\alato(#1){}
\def\veroparagrafo{\number\numsec}\def\veraformula{\number\numfor}
\def\Eq(#1){\eqno{\etichetta(#1)\alato(#1)}}
\def\eq(#1){\etichetta(#1)\alato(#1)}
\def\Eqa(#1){\eqno{\etichettaa(#1)\alato(#1)}}
\def\eqa(#1){\etichettaa(#1)\alato(#1)}
\def\equ(#1){\senondefinito{e#1}$\clubsuit$#1\else\csname e#1\endcsname\fi}
\let\EQ=\Eq


\def\\{\noindent}
\def\v{\vskip.1cm}
\def\vv{\vskip.2cm}


%
  \def\P{\mathop{\textrm{\rm P}}\nolimits}                  
  \def\d{\mathop{\textrm{\rm d}}\nolimits}                  
  \def\exp{\mathop{\textrm{\rm exp}}\nolimits}              
	\def\supp{\mathop{\textrm{\rm supp}}\nolimits}            
	\def\Int{\mathop{\textrm{\rm Int}}\nolimits}            
	\def\Ext{\mathop{\textrm{\rm Ext}}\nolimits}            
	\def\LRO{\mathop{\textrm{\rm LRO}}\nolimits}            
    \def\sf{\mathop{\textrm{\rm sf}}\nolimits}            
    \def\conv{\mathop{\textrm{\rm conv}}\nolimits}            

    \newcommand\bfblue[1]{\textcolor{blue}{\textbf{#1}}}
\newcommand\blue[1]{\textcolor{blue}{}}

\thispagestyle{empty}

\begin{center}
{\LARGE Contour methods for long-range Ising models:\\ weakening  nearest-neighbor interactions and adding decaying fields.}
\vskip.5cm
Rodrigo Bissacot$^{1}$, Eric O. Endo$^{1,2}$, Aernout C. D.  van Enter$^{2}$,\\ Bruno Kimura$^{3}$ and Wioletta M. Ruszel$^{3}$
\vskip.3cm
\begin{footnotesize}
$^{1}$Institute of Mathematics and Statistics (IME-USP), University of S\~{a}o Paulo, Brazil\\
$^{2}$Johann Bernoulli Institute, University of Groningen, the Netherlands\\
$^{3}$Delft Institute for Applied Mathematics, Technical University Delft, the Netherlands
\end{footnotesize}
\vskip.1cm
\begin{scriptsize}
emails: rodrigo.bissacot@gmail.com; eric@ime.usp.br; avanenter@gmail.com; bruno.hfkimura@gmail.com; w.m.ruszel@tudelft.nl
\end{scriptsize}

\end{center}

\def\be{\begin{equation}}
\def\ee{\end{equation}}

\vskip1.0cm
\begin{quote}
{\small

\textbf{Abstract.} \begin{footnotesize} We consider ferromagnetic long-range Ising models which display phase transitions.  They are one-dimensional Ising ferromagnets, in which the interaction is given by $J_{x,y} = J(|x-y|)\equiv \frac{1}{|x-y|^{2-\alpha}}$ with $\alpha \in [0, 1)$, in particular, $J(1)=1$. For this class of models one way in which one can prove the phase transition is via  a kind of Peierls contour argument, using the adaptation of the Fr\"ohlich-Spencer contours for $\alpha \neq 0$, proposed by Cassandro, Ferrari, Merola and Presutti. As proved by  Fr\"ohlich and Spencer for $\alpha=0$ and conjectured by Cassandro et al for the region they could treat, $\alpha \in (0,\alpha_{+})$ for $\alpha_+=\log(3)/\log(2)-1$, although in the literature dealing with contour methods for   these models it is generally  assumed that $J(1)\gg1$, we will show that this condition can be removed in the contour analysis. In addition, combining our theorem with a recent result of Littin and Picco we prove the persistence of the contour proof  of the phase transition for any $\alpha \in [0,1)$.  Moreover, we show that when we add a magnetic field decaying to zero, given by $h_x= h_*\cdot(1+|x|)^{-\gamma}$ and $\gamma >\max\{1-\alpha, 1-\alpha^* \}$ where $\alpha^*\approx 0.2714$, the transition still persists. \end{footnotesize}

}
\end{quote}

\section{Introduction}
\noindent

The rigorous study of phase transitions for one-dimensional Ising models with long-range slowly decaying  interactions (Dyson models) is a classical subject in one-dimensional  statistical mechanics.  One of the earliest highlights,  almost 50 years ago, was  Dyson's proof of a phase transition \cite{Dys1, Dys2, Dys3}, proving a conjecture due to Kac and Thompson \cite{KacT}. Long-range Ising models with slow polynomial decay, as well as the somewhat  related hierarchical models, have been called ``Dyson models" in the literature. We will mostly call our polynomially decaying models ``long-range Ising models'' but sometimes refer to them as ``Dyson models''.

The formal Hamiltonian of these models is given by:
\begin{equation}\label{hamiltonian}
H(\sigma)= - \sum_{\substack{x\neq y}}J_{x,y}\sigma_x \sigma_y - \sum_{x}h_x\sigma_x   
\end{equation}

Here the sites $x,y$ live in $\mathbb{Z}$, the $\sigma_x$ are Ising spins.   More precise definitions are given in the next section. We first mention what is known   for the zero-field case, i.e. when $h_x=0$ for all $x$.

 We consider ferromagnetic interactions $J_{x,y} \geq 0$ given by $J_{x,y}=|x-y|^{-2+\alpha}$ with $\alpha<1$. It is well known that for $\alpha<0$ there is no phase transition, and Dyson showed in \cite{Dys1} via comparison with a hierarchical model, that, for $\alpha\in (0,1)$, we have a phase transition at low temperature.  
 
 Afterwards different proofs were invented to show the transition. 
 One of them used Reflection Positivity \cite{FILS}. The method of infrared bounds offers an alternative to obtaining bounds on contour probabilities. In fact, the authors of  \cite{FILS} remark that they can cover a general class of long-range one-dimensional pair interactions, including the ones treated in \cite{Dys1}.
 
Shortly after, Fr\"ohlich and Spencer \cite{frsP} showed the existence of a phase transition for $\alpha=0$. The proof of these authors was done by a contour argument; they invented a notion of one-dimensional contours on $\mathbb{Z}$ in order to prove the phase transition. Their strategy more or less followed the classical Peierls contour argument used for the standard nearest-neighbor Ising model, but with a substantially more sophisticated definition of  contours. Phase transitions for larger $\alpha \in (0,1)$ can then be deduced by Griffiths inequalities for low enough temperature.

 Yet another way to derive the transition was a comparison with independent long-range percolation via Fortuin inequalities and Griffiths inequalities for the $\alpha=0$ case, as discussed in \cite{ACCN}. In that paper it was also shown that the transition for $\alpha=0$ is a hybrid one, in the sense that  the magnetisation is discontinuous and at the same time the energy is continuous as a function of temperature (Thouless effect). Moreover, for $\alpha=0$ it is known that there is a temperature interval below the transition temperature where the system is critical, in the sense that the covariance is nonsummable,  and at the same time the system is magnetized.  
  
 Cassandro et al. in \cite{CFMP} rigorously formalized the contour argument of \cite{frsP} in the parameter regime $0\leq \alpha < \alpha_+$, where $\alpha_{+}:=\log 3/\log 2-1 \approx 0.5849$. The construction allows a more precise description of various properties of the  model. It has been used in  various follow-up papers \cite{BEvELN, CMP, CMPR, COP1, COP2, Litthes,LP}. We should emphasize that, although  the use of contour arguments may look somewhat unwieldy in comparison with other approaches, it is much more robust. Indeed it has been used to analyze Dyson models in random \cite{COP1, COP2} and periodic fields \cite{Ke},   for interface behaviour and phase separation \cite{CMP, CMPR}, for entropic repulsion \cite{BEvELN}, and here for the model in  decaying magnetic fields, all problems where alternative methods appear to break down.

 See also \cite{Joh} for another, somewhat  related approach.\\

However, the adaptation proposed by Cassandro et al. in \cite{CFMP} needed the following technical assumptions:
(A1): $\alpha \in [0,\alpha^+)$ and (A2): $J(1)\gg 1$.

Even the case of $\alpha =0$, previously obtained by Fr\"ohlich and Spencer, needs $J(1) \gg 1$ in the adaptation proposed by them. The intuition behind the condition is more or less clear; it  makes the model closer to a nearest-neighbor interaction model where, in principle, contour arguments  might work more easily. Despite the condition being rather  artificial  and proof-generated, the constraint asking for $J(1) \gg 1$ is present in many later papers about Dyson models and the proof presented in \cite{CFMP} depends strongly on this hypothesis. 

As regards the restriction on $\alpha$, Littin in his thesis \cite{Litthes}, and then Littin and Picco \cite{LP}, showed that, using quasi-additive properties of the Hamiltonian of the corresponding contour model and applying the results from \cite{CFMP}, one can modify the contour argument so that it implies the phase transition for all $\alpha\in [0, 1)$. Due to the fact that the authors in \cite{LP} use energetic lower bounds from \cite{CFMP} which assume large nearest-neighbour interaction $J(1)$, they still use assumption (A2) in their arguments.

Our motivation for the present work is two-fold: first we want to present an argument to remove assumption (A2) for the zero-field case and secondly  we want to show persistence of a phase transition for one-dimensional long-range models in the presence of  external fields decaying to zero at infinity with a power $\gamma$, in particular, for fields given by $h_x= h_*(1+ |x|)^{-\gamma}$ and $1-\alpha < \gamma$. More precisely, our results combined with existing results imply that there is a trade-off between the restricting the parameter range of $\gamma$ to $\gamma> \max \{ 1-\alpha, 1-\alpha^* \}$ and general $J(1)$  and assuming $J(1)\gg 1$ and choosing  $\gamma> \max \{ 1-\alpha, 1-\alpha_+\}$ where $\alpha_+ > \alpha^*$ will be specified later. Note that our results apply to the latter case as well.

Before describing the 
rest of the paper, 
 we will discuss briefly the context of these results with respect to the  hypotheses and technicalities of the proof. Let us mention that a short announcement of some of our results, but without rigorous proofs, is contained in \cite{BEvEKLNR}.
\vspace*{0.05cm}

Considering the first result in the \textit{zero-field case}, although proofs for the existence of a phase transition were known, our estimates allow firstly to drop the (A2) assumption, and then, by using monotonicity of the Hamiltonian with respect to $\alpha$, we are also able to remove the first assumption (A1). 

As regarding the \textit{decaying-field case} we know that phase transitions for non-zero  uniform fields are forbidden due to the Lee-Yang circle theorem \cite{LY}. 

The heuristics behind the inequality $1-\alpha  < \gamma$ can be obtained  as follows. We observe  that the contribution of the interaction of a finite interval $\Lambda$ with its complement is of order  $O(|\Lambda|^{\alpha})$, whereas the contribution from the external field is of order $O(|\Lambda|^{1-\gamma})$. 

We now compare  the exponents. If the interaction energy dominates the field energy for large $\Lambda$, a contour argument has a chance of working. \\
This intuition is also what is underlying Imry-Ma arguments for analyzing the stability of phase transitions in the presence of  random fields. It has been confirmed for decaying fields in higher-dimensional nearest-neighbour models, see below. \\ 

It can also be applied to a decaying field the strength of which decays with power $\gamma$ but which has  random signs. In this case the field energy behaves like $O(|\Lambda|^{\frac{1}{2} - \gamma})$. This case has also been considered before by J. Littin (private communication) \cite{Lit2}. We note that the case $\gamma =0$ reduces to the known Imry-Ma analysis as presented in \cite{COP1,COP2}.\\   

Note that the analogous question of the persistence of phase transitions in decaying fields already was studied before in some short-range models, see  \cite{BC, BCCP, BEvE,CV}.

The paper is organized as follows. In the section \ref{notation} we introduce the definition of the model and some notation. The following section \ref{mainR} presents the main results. Section \ref{triangle} introduces the construction of the  contours and section \ref{proof} contains the proofs of the main theorems including the Peierls argument. Finally section \ref{conc} concludes with a summary and discussion about open questions.





\section{Notation}\label{notation}

Let $\Omega=\{-1,1\}^{\mathbb{Z}}$ be the set of configurations $\sigma=(\sigma_x)_{x\in \mathbb{Z}}$ on $\mathbb{Z}$. The \emph{Hamiltonian} in a finite volume $\Lambda$ with uniform boundary condition $\omega$ (either ``plus" or ``minus" configurations) can be rewritten in the following way
\be \label{Hamil}
H^{\omega}_{\Lambda,\bar{h}, \alpha}(\sigma)=\frac{1}{2}\sum_{\substack{(x,y) \in \Lambda\times \Lambda}}J(|x-y|)\mathbbm{1}_{\sigma_x\neq \sigma_y} + \sum_{\substack{x\in \Lambda \\ y\notin \Lambda}}J(|x-y|)\mathbbm{1}_{\sigma_x\neq \omega_y} + \sum_{x\in \Lambda}h_x\mathbbm{1}_{\sigma_x = -1},
\ee
where the \emph{coupling constants} $J_{x,y}=J(|x-y|)$ are defined, as already mentioned, by
\be
J(|x-y|)=
\begin{cases}
J &\text{ if }|x-y|=1;\\
|x-y|^{-2+\alpha} &\text{ if }|x-y|>1,
\end{cases}
\ee
where $J(1)=J>0$ and $0\leq \alpha < 1$, and the \emph{external field} $\bar{h}=(h_x)_{x\in \mathbb{Z}}$ is defined by 
\be
h_x=h_*\cdot(1+|x|)^{-\gamma}
\ee
with $\gamma>0$ and $h_* \in \mathbb{R}$. In our theorem $J(1)=1$, but we keep this separation between the nearest-neighbour term and the other ones for historical reasons, and also because  controlling the nearest-neighbour term is one of the main contributions of our paper.

For \emph{inverse temperature} $\beta>0$, the associated \emph{Gibbs measure} in $\Lambda \subset \mathbb{Z}$ with boundary condition $\omega$ is given by
\be
\mu^{\omega}_{\Lambda,\bar{h},\beta}(\sigma)=\frac{e^{-\beta H^{\omega}_{\Lambda,\bar{h}, \alpha}(\sigma)}}{Z^{\omega}_{\Lambda,\bar{h},\beta}},
\ee
where $Z^{\omega}_{\Lambda,\bar{h},\beta}$ is the \emph{partition function}
\be
Z^{\omega}_{\Lambda,\bar{h},\beta} = \sum_{\sigma \in \Omega}e^{-\beta H^{\omega}_{\Lambda,\bar{h}, \alpha}(\sigma)} 
\ee
which contains a sum over all configurations $\sigma \in \Omega$ satisfying $\sigma_x=\omega_x$ for every $x\notin \Lambda$.

We denote by $\mu^{+}_{\Lambda,\bar{h},\beta}$ (resp. $\mu^{-}_{\Lambda,\bar{h},\beta}$)  the Gibbs measure with plus (resp. minus) boundary condition, i.e., $\omega_x=+1$ (resp. $\omega_x=-1$) for every $x\in \mathbb{Z}$.
For a sequence $(\Lambda_n)_{n\ge 1}$ of finite sets in $\mathbb{Z}$, we write $\Lambda_n\uparrow \mathbb{Z}$ if for every $x\in \mathbb{Z}$, there exists $n_x\ge 1$, such that $x\in \Lambda_n$ for every $n\ge n_x$. We know that, for every  $(\Lambda_n)_{n\ge 1}$ with $\Lambda_n\uparrow \mathbb{Z}$ and any local function $f$
\be
\lim_{\Lambda_n\uparrow \mathbb{Z}}\mu^{+}_{\Lambda_n,\bar{h},\beta}(f)=\mu^+_{\bar{h},\beta}(f) \quad \text{ and } \quad
\lim_{\Lambda_n\uparrow \mathbb{Z}}\mu^{-}_{\Lambda_n,\bar{h},\beta}(f)=\mu^-_{\bar{h},\beta}(f). 
\ee
Let $\mathcal{F}$ be the sigma-algebra over $\Omega$ generated by the cylinder sets, and let $\mathcal{M}(\Omega,\mathcal{F})$ be the set of all probability measures on $(\Omega, \mathcal{F})$.  We denote the set of Gibbs measures by
\be
\mathcal{G}(\bar{h},\beta)=\overline{\conv}\left\{ \mu\in \mathcal{M}(\Omega,\mathcal{F}): \text{ there exist }(\Lambda_n)_{\geq 1} \text{ and }(\omega_n)_{n\geq 1} \text{ s.t. } \lim_{\Lambda_n\uparrow \mathbb{Z}}\mu^{\omega_n}_{\Lambda_n,\bar{h},\beta}=\mu \right\}.
\ee
Note that $\mu^+_{\bar{h},\beta}$ and $\mu^-_{\bar{h},\beta}$ are in $\mathcal{G}(\bar{h},\beta)$.
We say that the model \emph{undergoes a phase transition at $(\bar{h},\beta)$} if $|\mathcal{G}(\bar{h},\beta)|> 1$. This definition is in fact equivalent to showing that, for the same $(\bar{h},\beta)$, we have $\mu^+_{\bar{h},\beta} \neq \mu^-_{\bar{h},\beta}$ (this equivalence follows from the FKG inequality). We say that the model has \emph{uniqueness at $(\bar{h},\beta)$} if $|\mathcal{G}(\bar{h},\beta)|=1$.

\section{Main Results}\label{mainR}

In this section we present our main results. The first result concerns the phase transition for the Dyson model for any $\alpha \in [0,1)$ and removing the $J(1)\gg1$ assumption and the second result concerns the persistence of the phase transition under decaying external fields. 
\begin{theorem}\label{thm1}
Let us consider an Ising model on $\mathbb{Z}$ with Hamiltonian given by \eqref{Hamil} for $\alpha \in [0,1)$, $J(1) = 1$  and $h_x\equiv 0$ for all $x\in \mathbb{Z}$. Then there exists $\beta_c >0$ such that for all $\beta > \beta_c$ we have a convergent low-temperature expansion which implies that $|\mathcal{G}(0,\beta)|> 1$.
\end{theorem}

\begin{theorem}\label{thm2}
Let $\alpha^*$ be such that $\sum_{k=1}^{\infty} k^{-2+\alpha^*}=2$.
Consider an Ising model on $\mathbb{Z}$ with  Hamiltonian given by \eqref{Hamil} such that $\overline{h}=(h_x)_{x\in \mathbb{Z}}$ are defined by $h_x=h_* \cdot (1+|x|)^{-\gamma}$.  We assume either
\begin{itemize}
\item $\alpha \in (0,1)$, $J(1) = 1$, $h_*\in \mathbb{R}$ and $\gamma > \max\{1-\alpha, 1-\alpha^*\}$, or  
\item $\alpha \in [0,\alpha^*)$, $J(1) = 1$, $\gamma = 1-\alpha$ and $h_*$ small enough.
\end{itemize}
Then there exists $\beta_c >0$ such that for all $\beta > \beta_c$ we have that $|\mathcal{G}(\overline{h},\beta)|> 1$.
\end{theorem}

{\bf Remark:} In the case of $\alpha=0$, although it  is the most complicated in general (the proof of the phase transition took more time \cite{frsP}, there is a mixed first-order second-order transition -the Thouless effect, as proven in \cite{ACCN}, see also \cite{Ra}- and there exists an intermediate phase, magnetised but with nonsummable covariance \cite{IN} in some temperature interval below the transition temperature), for our problem the situation becomes  actually  somewhat simpler. Indeed, in that case the condition  $\gamma > 1$ implies that the field energy is uniformly bounded, and thus the phase transition persists, whenever it occurs, due to the arguments of \cite{BC}. This applies at all temperatures where there is a phase transition, including the intermediate phase, and is not dependent 
on the applicability of contour arguments. However, for the critical value $\gamma=1$ and the field strength weak enough, our proof does apply only at very low temperatures.
\section{Triangles and contours}\label{triangle}

The proof will consist of a Peierls type argument in one dimension. 
We start by defining the notion of triangles which was first described in \cite{frsP} for $\alpha=0$ and adapted for $0\leq \alpha< 1$ in \cite{CFMP}.

Let $N\ge 1$, and consider an interval $\Lambda=\Lambda_N=[-N,N]$. Define the dual lattice $\Lambda^*=\Lambda+\frac{1}{2}$ as the set  $\Lambda$ shifted by $1/2$. Given a configuration $\sigma\in \{-1,+1\}^{\Lambda}$, let us define configurations of triangles. A \emph{spin -flip point} is a site $i$ in $\Lambda^*$ such that $\sigma_{i-\frac{1}{2}}\neq \sigma_{i+\frac{1}{2}}$. For each spin-flip point $i$, let us consider the interval $\left[i-\frac{1}{100},i+\frac{1}{100}\right]\subset \mathbb{R}$ and choose a real number $r_i$ in this interval such that, for every four distinct $r_{i_l}$ with $l=1,\ldots,4$, we have $|r_{i_1}-r_{i_2}|\neq |r_{i_3}-r_{i_4}|$. The $r_i$ are the bases of the triangles, and the last condition on $r_i$ is to avoid ambiguity in the construction of the triangles below.

For each spin-flip point $i$, we start growing a ``$\lor$-line'' at $r_i$ where this $\lor$-line is embedded in $\mathbb{R}^2$ with angles $\pi/4$ and $3\pi/4$. If at some time two $\lor$-lines starting from different spin-flip points touch, the other two lines starting from those two spin-flip points stop growing and we remove those lines -which did not form a triangle-, and  we keep continuing this process. The process can also be seen in the following way: for each $r_i$, we draw a straight vertical line passing through it. Take the smallest distance between these lines, let us call the corresponding $r_i$ and $r_j$ the spin-flip points of these lines, and draw a isosceles triangle with base angle $\pi/4$. Then, remove the lines associated to $r_i$ and $r_j$, and continue the process.

Note that, for any finite interval $\Lambda$ with homogeneous boundary condition, the number of spin flips is even, and so every $r_{i}$ is a vertex of some triangle. Let us denote by $\mathcal{S}_{\Lambda^+}$ be the set of configurations with plus boundary condition, i.e., for $\sigma\in \mathcal{S}_{\Lambda^+}$ we have $\sigma_x=+1$ for every $x\notin \Lambda$.

We denote a triangle by $T$, and introduce the following notations,
$$
\begin{aligned}
\{x_{-}(T),x_{+}(T)\}&=\text{ the left and right root of the associated $\lor$-lines, respectively},\\
\Delta(T)&=[x_{-}(T),x_{+}(T)]\cap \mathbb{Z}, \text{ is the base of the triangle }T,\\
|T|&=|\Delta(T)|, \text{ be the mass of the triangle }T,\\
\sf^*(T)&=\left\{ \inf \Delta(T)-\frac{1}{2},\sup \Delta(T) +\frac{1}{2} \right\},
\end{aligned}
$$
and $\mathbb{Z}$ is equipped with the natural order
\be
\d(T,T')=\d\{\sf^*(T),\sf^*(T')\}.
\ee
By definition of the triangles, for every pair of triangles $T\neq T'$,
\begin{equation}\label{eq:triangle_condition}
\d(T,T')\ge \min\{|T|,|T'|\}.
\end{equation}
We denote by $\mathcal{T}_{\Lambda^+}$ be the set of configurations of triangles $\underline{T}=\{T_1,\ldots,T_n\}$ satisfying (\ref{eq:triangle_condition}) and such that $\Delta(T_i)\subset \Lambda$ for every $i=1\ldots,n$. Given $\sigma\in \mathcal{S}_{\Lambda^+}$, we denote by $\underline{T}(\sigma)$ the configuration of triangles constructed from the configuration $\sigma$.
A family of triangles $\underline{T}$ is a \emph{compatible family} if there exists a configuration $\sigma\in\mathcal{S}_{\Lambda^+}$ such that $\underline{T}=\underline{T}(\sigma)$.

\begin{definition}\label{def:contour}
Let $c>1$ be a positive real number and $\underline{T}\in \mathcal{T}_{\Lambda^+}$ be a compatible configuration of triangles, then a \emph{configuration of contours} $\underline{\Gamma}\equiv \underline{\Gamma}(\underline{T})$ is a partition of $\underline{T}$ whose elements, called \emph{contours}, are determined by the following properties

\textbf{P.0:} Let $\underline{\Gamma}\equiv (\Gamma_1,\ldots,\Gamma_N)$, $\Gamma_i=\{T_{m,i}:1\le m\le k_i\}$, then $\underline{T}=\{T_{m,i}:1\le m\le k_i,1\le i\le N\}$.

\textbf{P.1:} Contours are well-separated from each other. Consider the base of a contour $\Gamma$ by
\be
\Delta(\Gamma)=\bigcup_{T\in \Gamma}\Delta(T).
\ee 
Any pair $\Gamma\neq \Gamma'$ in $\underline{\Gamma}$ verifies one of the following two alternatives

\textnormal{(1)} $\Delta(\Gamma_i)\cap \Delta(\Gamma_j)=\emptyset$.

\textnormal{(2)} Either $\Delta(\Gamma_i)\subseteq \Delta(\Gamma_j)$ or $\Delta(\Gamma_j)\subseteq \Delta(\Gamma_i)$. Moreover, supposing that the first case is satisfied, then for any triangle $T_{m,j} \in \Gamma_j$, either $\Delta(\Gamma_i) \subseteq T_{m,j}$ or $\Delta(\Gamma_i)\cap T_{m,j}=\emptyset$.

In both cases (1) and (2),
\be
\d(\Gamma,\Gamma'):=\min_{\substack{T\in \Gamma \\ T'\in\Gamma'}}\d(T,T')>c\min\{|\Gamma|,|\Gamma'|\}^3,
\ee
where
\be
|\Gamma|=\sum_{T\in \Gamma}|T|.
\ee

\textbf{P.2:} Independence. Let $\{\underline{T}^{(1)},\ldots, \underline{T}^{(k)}\}$ be configurations of triangles; consider the contours of the configuration $\underline{T}^{(i)}$ by $\underline{\Gamma}(\underline{T}^{(i)})= \{\Gamma_j^{(i)}: j=1,\ldots,n_i\}$. If for any distinct pair $\Gamma_{j}^{(i)}$ and $\Gamma_{j'}^{(i')}$ the property P.1 is satisfied, then
\begin{equation}
\underline{\Gamma}\left(\underline{T}^{(1)},\underline{T}^{(2)},\ldots, \underline{T}^{(k)}\right) = \left\{ \Gamma_j^{(i)}:i=1,\ldots,k; j=1,\ldots, n_i \right\}.
\end{equation}
\end{definition}
The proof of existence and uniqueness of an algorithm that produces $\underline{\Gamma}$ satisfying Definition \ref{def:contour} is given in \cite{CFMP}.

Note that it is straightforward that there exists a bijection between spin configurations in $\mathcal{S}_{\Lambda^+}$ and triangles in $\mathcal{T}_{\Lambda^+}$, and also $\underline{T} \mapsto \underline{\Gamma}(\underline{T})$ is a bijection, where $\underline{T}\in\mathcal{T}_{\Lambda^+}$. Thus, there exists a bijection between spin configurations in $\mathcal{S}_{\Lambda^+}$ and contour configurations.

\begin{definition}
We say that a family of contours $\{\Gamma_0,\Gamma_1,\ldots,\Gamma_n\}$ is \emph{compatible} if there exists $\sigma_{\Lambda}\in \mathcal{S}_{\Lambda^+}$ such that $\underline{\Gamma}=\underline{\Gamma}(\underline{T}(\sigma_{\Lambda}))$.
\end{definition}

The expression of the Hamiltonian with plus boundary condition of a family of compatible contours $\underline{\Gamma}$ is given by
\be
H^+_{\alpha,\bar{h}}(\underline{\Gamma})=\frac{1}{2}\sum_{\substack{x,y\in \mathbb{Z}\\ x\neq y}}J(|x-y|)\mathbbm{1}_{\sigma_x(\underline{\Gamma})\neq \sigma_y(\underline{\Gamma})} +\sum_{x\in \mathbb{Z}}h_x\mathbbm{1}_{\sigma_x(\underline{\Gamma})=-1},
\ee
where $\sigma_x(\underline{\Gamma})$ is the spin at vertex $x$ in the presence of $\underline{\Gamma}$. We write $H^+_{\bar{h}}:=H^+_{\alpha,\bar{h}}$ when it is possible to omit $\alpha$ without ambiguity. We denote by $H^+_{\alpha}(\underline{\Gamma})$ when we have absence of external fields, i.e.,
\be\label{H0}
H^+_{\alpha}(\underline{\Gamma})=\frac{1}{2}\sum_{\substack{x,y\in \mathbb{Z}\\ x\neq y}}J(|x-y|)\mathbbm{1}_{\sigma_x(\underline{\Gamma})\neq \sigma_y(\underline{\Gamma})}.
\ee
Let us mention at this point an important property which we will use later, namely the monotonicity of the Hamiltonian $H^+_{\alpha}$ as a function of $\alpha$. More precisely
\begin{equation}\label{mono} 
\text{for } \alpha \geq \alpha' \Rightarrow H^+_{\alpha}(\underline{\Gamma} )\geq   H^+_{\alpha'}(\underline{\Gamma})
\end{equation}
for any contour configuration $\underline{\Gamma}$ and $H^+_{\alpha}$ defined in \eqref{H0}.
For $L\ge 1$, let us consider the function $W_{\alpha}$ given by
\begin{equation}
W_{\alpha}(L) = \sum_{x=1}^{L} \left[\sum_{\substack{y \in [L+1,2L] \cap \mathbb{Z}\\ y \in [-L+1,0] \cap \mathbb{Z}}}J(|x-y|) - \sum_{\substack{y \in [2L+1,\infty)\cap \mathbb{Z} \\ y \in (-\infty,-L] \cap \mathbb{Z}}}J(|x-y|)	\right].	
\end{equation}

Given $\alpha \in [0,1)$, let  $c$,  the constant from property P.1,  be large enough. By \cite{CFMP}, using the monotonicity of $J(\cdot)$, given a contour $\Gamma_0$, we have
\be
H_{\alpha}^+(\Gamma_0)\ge \sum_{T\in \Gamma_0}W_{\alpha}(|T|),
\ee
and given a configuration of contours $\underline{\Gamma}$, for any $\Gamma_0 \in \underline{\Gamma}$ we have
\be
H_{\alpha}^+(\Gamma_0|\underline{\Gamma}\backslash \{\Gamma_0 \})\geq \frac{1}{2}\sum\limits_{T \in \Gamma_0}  W_{\alpha}(|T|),
\ee
where
\be\label{boundhamiltonian}
H_{\alpha}^+(\Gamma_0|\underline{\Gamma}\backslash \{\Gamma_0\})=H_{\alpha}^+(\underline{\Gamma}) - H_{\alpha}^+(\underline{\Gamma}\setminus \{ \Gamma_0\}).
\ee
\subsection*{Entropy of Contours}  For any kind of Peierls argument \cite{Pe} where a notion of a contour is used to obtain the phase transition for a model, the entropy  of the contours (that is the -logarithm of- the number of contours of a given size) should be controlled. Recently, generating functions were used on this combinatorial problem of counting contours associated to short-range models on regular trees \cite{ABE} and on $\mathbb{Z}^d, d\geq2$, see \cite{BB}. For trees the method allows us to find the exact number of contours of a fixed size and for $\mathbb{Z}^d$ these are the best estimates until now.

The entropy of contours introduced in \cite{CFMP} was inspired by \cite{frsP} and it satisfies the following estimates:

For any real $b>0$ large enough and  integer $m\ge 1$, we have, for $\alpha\in (0,1)$,
\be\label{entropy1}
\sum_{\substack{|\Gamma|=m\\ 0\in \Gamma}}e^{-b\sum_{T\in \Gamma}|T|^{\alpha}}\le 2me^{-bm^{\alpha}}.
\ee
In the case $\alpha=0$,
\be\label{entropy2}
\sum_{\substack{|\Gamma|=m\\ 0\in \Gamma}} e^{-b\sum_{T\in \Gamma}(\log (|T|)+4)}\le 2m e^{-b (\log m+4)}.
\ee

The proof is done by induction on $m$ and using a graphical representation of contours by trees. This representation requires a process, called the \emph{square process}, to create the structure of a tree from any contour. The existence of such  an algorithm can be found in \cite{CFMP}. 

\subsection*{Quasi-additive properties of the Hamiltonian}
In order to control the energy for the proof of phase transition via contours, Cassandro et al. in \cite{CFMP} showed that there exist $C_{\alpha}>0$ for $\alpha \in (0,\alpha_+)$ such that
\be \label{LBH}
H_{\alpha}^+(\Gamma_0|\underline{\Gamma}\backslash \{ \Gamma_0 \})\ge \frac{C_{\alpha} }{2}\sum_{T\in \Gamma_0}|T|^{\alpha},
\ee
and
\be \label{LBH1}
H_{\alpha}^+(\Gamma_0|\underline{\Gamma}\backslash \{ \Gamma_0 \})\ge \frac{C_{0}}{2}\sum_{T\in \Gamma_0}\log(|T|+4),
\ee
for $\alpha=0$. The constant $C_{\alpha}=\frac{3-2^{1+\alpha}}{\alpha(1-\alpha)}$ converges to zero when $\alpha \to \alpha_{+}$. Those bounds only work when the energy of the nearest neighbors $J(1)$ is large enough. Littin and Picco \cite{LP} showed that there is no possibility to extend those bounds in terms of triangles to any $\alpha \in [0,1)$, constructing a sequence of contours $(\Gamma_n)_{n\ge 1}$ such that
\be
\lim_{n\to \infty} \frac{H_{\alpha}^+(\Gamma_n)}{\sum_{T\in \Gamma_n}|T|^{\alpha}}=0
\ee
for $0\leq \alpha < \alpha_{+}$. 

Their construction explains, in some sense, the value of the constant $\alpha_+=\log 3/\log 2-1$ which could be considered  kind of magical, since the Hausdorff dimension of the Cantor set is exactly $\log 2/\log 3$. In fact, a discrete Cantor set appears in the construction. However, \cite{LP} showed the quasi-additive property of the Hamiltonian for any $\alpha \in [0,1)$, in other words, the following inequality is satisfied,
\be\label{littin}
H_{\alpha}^+(\Gamma_0|\underline{\Gamma}\backslash \{ \Gamma_0 \})\ge K_c(\alpha)H_{\alpha}^+[\Gamma_0],
\ee
where $K_{c}(\alpha)=1-\frac{\alpha}{c^{\alpha-1}}-\frac{\pi^2}{6c}$ satisfies $0<K_c(\alpha)\le \frac{1}{2}$ where $c$ is large enough. This property allows the authors to prove the phase transition for any $\alpha \in [0,1)$ as a corollary from the case when $\alpha \in [0, \alpha_{+})$.  As \cite{LP} mentioned, Fr\"{o}hlich and Spencer \cite{frsP} already noticed that the Hamiltonian satisfies the quasi-additivity.


\section{Proof of the main Theorems}\label{proof}

\subsection{Proof of Theorem \ref{thm1}}


By Lemma A.1 from \cite{CFMP} there is a lower bound for $W_{\alpha}$ which allows one to use the contour argument and prove the phase transition. In fact, they showed that, for any $\alpha \in [0,\alpha_{+})$ and $J(1)$ large enough, we have, for every $L\ge 1$,
\be
W_{\alpha}(L)\ge
\begin{cases}
\zeta_{\alpha} L^{\alpha}, &\text{ if }\alpha \in (0,\alpha_+);\\
2(\log L+4), &\text{ if }\alpha=0,
\end{cases}
\ee
The proof in \cite{CFMP} uses monotonicity of $J(\cdot)$ to replace sums by integrals, and $J(1)$ should be large if we desire to have this bound for small $L$, this is an essential hypothesis for the argument. The proposition below shows that, at the cost of reducing the interval of $\alpha$, the condition $J(1)\gg 1$ can be substituted by $J(1)=1$.

\begin{proposition}\label{main2}
Consider $0<\alpha^{*}<1$ satisfying
\be
\sum_{n=1}^{\infty}\frac{1}{n^{2-\alpha^{*}}}=2.
\ee
Let $\alpha \in [0,\alpha^{*})$ with $\alpha^*\approx 0.2714$. Then there is a constant $\zeta_{\alpha}>0$ such that 
\begin{equation}\label{lowerbound}
W_{\alpha}(L) \geq \zeta_{\alpha} \chi_{\alpha}(L)		
\end{equation}	
holds for all $L \geq 1$, where
\be\label{chi}
\chi_{\alpha}(L) =
\begin{cases}
L^{\alpha} &\text{if $\alpha \in (0,\alpha^{*})$, and}\\
\log L+4 & \text{if $\alpha =0$.}	
\end{cases}	
\ee
Thus, let the constant $c$ in the definition of the contours be large enough. For any contour $\Gamma \in \underline{\Gamma}$, we have
\be
H_{\alpha}^+(\Gamma) \ge \zeta_{\alpha}\lVert \Gamma \rVert_{\alpha} \quad \text{ and }\quad
H_{\alpha}^+[\Gamma|\underline{\Gamma}\backslash \{ \Gamma \}]\ge \frac{\zeta_{\alpha}}{2}\lVert \Gamma \rVert_{\alpha},
\ee
where $\lVert \Gamma \rVert_{\alpha}=\sum_{T\in \Gamma}\chi_{\alpha}(|T|)$.
\end{proposition}

\begin{proof}
For every $L\ge 1$, we can write $W_{\alpha}(L)$ as
\be\label{eqW1}
W_{\alpha}(L)
=2\sum_{x = 1}^{L}\sum_{y = L+1-x}^{2L-x} \frac{1}{y^{2-\alpha}} - 2\sum_{x = 1}^{L}\sum_{y = 2L+1-x}^{\infty} \frac{1}{y^{2-\alpha}}.
\ee
Splitting the first term of the equation (\ref{eqW1}) into $y\in [L+1-x,L]\cup [L+1,2L-x]$, and commuting the order of the sums in both terms, we find
\be
\begin{split}
W_{\alpha}(L) 
& = 2\sum_{x = 1}^{L}\sum_{y = L+1-x}^{L} \frac{1}{y^{2-\alpha}} + 2\sum_{x = 1}^{L-1}\sum_{y = L+1}^{2L-x} \frac{1}{y^{2-\alpha}}  - 2\sum_{x = 1}^{L}\sum_{y = 2L+1-x}^{\infty} \frac{1}{y^{2-\alpha}} \\
&=  2\sum_{y = 1}^{L} y \frac{1}{y^{2-\alpha}} + 4\sum_{y = L+1}^{2L-1} (2L-y) \frac{1}{y^{2-\alpha}}  - 2L \sum_{y = L+1}^{\infty} \frac{1}{y^{2-\alpha}} \\
&=2\sum_{y = 1}^{L} \frac{1}{y^{1-\alpha }} -4 \sum_{y = L+1}^{2L-1} \frac{1}{y^{1-\alpha }} + 8L\sum_{y = L+1}^{2L-1} \frac{1}{y^{2-\alpha}} - 2L \sum_{y = L+1}^{\infty} \frac{1}{y^{2-\alpha}}.
\end{split}
\ee
Given a real number $k$ and a positive integer $n$, let us define the number $H^{(k)}_{n}$ by	
\begin{equation}
H^{(k)}_{n} = \sum_{y=1}^{n}\frac{1}{y^{k}}.
\end{equation}
In particular, if $k=1$, we denote $H^{(1)}_{n}$ simply as $H_{n}$. Thus,
\be
\begin{aligned}
W_{\alpha}(L) &=2\left(3H_{L}^{(1-\alpha )} - 2 H_{2L-1}^{(1-\alpha)}-\frac{4}{L^{1-\alpha }}\right) + 8L\sum_{y=L}^{2L-1}\frac{1}{y^{2-\alpha}} - 2L \sum_{y=L+1}^{\infty} \frac{1}{y^{2-\alpha}}\\
&\geq 2\left(\frac{3}{L^{\alpha}}H_{L}^{(1-\alpha )} - \frac{2}{L^{\alpha}}H_{2L-1}^{(1-\alpha )}-\frac{4}{L}\right)L^{\alpha} + \frac{2}{1-\alpha }\left(3-2^{1+\alpha}\right)L^{\alpha}.
\end{aligned}
\ee
Suppose that $\alpha \in (0,\alpha^{*})$. Using the fact that
\begin{equation}
\int_{1}^{L+1}\frac{1}{x^{1-\alpha}}dx \leq H_{L}^{(1-\alpha)} \leq 1 + \int_{1}^{L}\frac{1}{x^{1-\alpha}}dx     
\end{equation}
and $\alpha^{*} < \alpha_{+}$, we have 
\be
\lim_{L \to \infty}\left(\frac{3}{L^{\alpha}}H_{L}^{(1-\alpha )} - \frac{2}{L^{\alpha}}H_{2L-1}^{(1-\alpha)}-\frac{4}{L}\right) = \frac{1}{\alpha}(3-2^{1+\alpha}) >0.	
\ee
Therefore, there exists $L_1\ge 1$ such that, for every $L> L_1$, we conclude
\be\label{L1}
W_{\alpha}(L) \geq \zeta^*_{{\alpha}} L^{\alpha},
\ee
where $\zeta^*_{\alpha}=2 (3-2^{1+\alpha})/ (1-\alpha)$. For  $\alpha = 0$, the quantity $W_{0}(L)$ satisfies the following inequality,
\be
W_{0}(L) \geq 2\left(3H_{L}-2H_{2L-1}-\frac{4}{L}+1\right).
\ee
Since
\be
\lim_{L \to \infty} \frac{1}{\log L+4}\cdot \left(3H_{L}-2H_{2L-1}-\frac{4}{L}+1\right)=1,
\ee
there exists $L_2\ge 1$ such that, for every $L> L_2$,
\be\label{L2}
W_{0}(L) \geq \log L +4.
\ee
In order to obtain lower bounds for $W_{\alpha}(L)$ for all $L$, it suffices to show that $W_{\alpha}(L)$ is positive for each $L$. First, note that
\begin{equation}
W_{\alpha}(1) = 2\left(2-\sum_{y=1}^{\infty}\frac{1}{y^{2-\alpha}}\right).	
\end{equation}
Since $\alpha \in [0,\alpha^{*})$, we have $W_{\alpha}(1) > 0$. Let us show that $W_{\alpha}$ is an increasing function with respect to $L$. 
Note that $W_{\alpha}(L)$ can be expressed as
\begin{equation}\label{eqW2}
W_{\alpha}(L) = 6H_{L}^{(1-\alpha )}-4H_{2L-1}^{(1-\alpha )}+8LH_{2L-1}^{(2-\alpha)}-6LH_{L}^{(2-\alpha)}-2L\sum_{y=1}^{\infty}\frac{1}{y^{2-\alpha}}.	
\end{equation}
Define $\Delta W_{\alpha}(L) := W_{\alpha}(L+1)-W_{\alpha}(L)$. Using (\ref{eqW2}), we have
\be
\Delta W_{\alpha}(L) 
=\frac{6}{(2L)^{2-\alpha}} +\frac{4}{(2L+1)^{2-\alpha}} + 6\sum_{y=L+1}^{2L-1}\frac{1}{y^{2-\alpha}} -2\sum_{y=2L+1}^{\infty}\frac{1}{y^{2-\alpha}}.
\ee
Note that $\Delta W_{\alpha}(1) \geq W_{\alpha}(1) >0$ and $\Delta W_{\alpha}(2) \geq W_{\alpha}(1)>0$. For $L\ge 3$,
\be
\begin{aligned}
\Delta W_{\alpha}(L) &\geq 6\int_{L+1}^{2L}\frac{1}{z^{2-\alpha}}dz -2\int_{2L}^{\infty}\frac{1}{z^{2-\alpha}}dz \\
&>\frac{2}{1-\alpha }\left[3\left(\frac{4}{3}\right)^{-1+\alpha} - 2^{1+\alpha}\right]L^{-1+\alpha}.
\end{aligned}
\ee
Note that the right-hand side of the equation above is positive whenever $\alpha < \overline{\alpha} \equiv \frac{\log (8/9)}{\log (2/3)} \approx 0.2904$. Since $\alpha^{*} < \overline{\alpha}$, we conclude that $\Delta W_{\alpha}(L)>0$ for every $L\ge 1$. For $\alpha \in (0,\alpha^{*})$, define
\be
\zeta_{\alpha}=\min\left\{W_{\alpha}(1), \frac{W_{\alpha}(2)}{2^{\alpha}},\ldots, \frac{W_{\alpha}(L_1)}{L_1^{\alpha}},\zeta_{\alpha}^*\right\},
\ee
and for $\alpha=0$, define
\be
\zeta_{0}=\min\left\{\frac{W_{0}(1)}{4}, \frac{W_{0}(2)}{\log 2 +4},\ldots, \frac{W_{0}(L_2)}{\log L_2 +4},1\right\}.
\ee
Thus $\zeta_{\alpha}>0$ for every $\alpha\in [0,\alpha^*)$. By (\ref{L1}) and (\ref{L2}), we conclude the result.
\end{proof}

The previous proposition yields the following corollary which can be easily proven using the monotonicity of $H^+_{\alpha}$ in terms of $\alpha$, see \eqref{mono}.
\begin{corollary}\label{coro1}
For all $\alpha \in [\alpha^*,1)$ there exists $\alpha' \in [0,\alpha^*)$ such that
for any contour $\Gamma \in \underline{\Gamma}$, we have
\be
H_{\alpha}^+(\Gamma) \ge \zeta_{\alpha'}\lVert \Gamma \rVert_{\alpha'} \quad \text{ and }\quad
H_{\alpha}^+[\Gamma|\underline{\Gamma}\backslash \Gamma]\ge \frac{\zeta_{\alpha'}}{2}\lVert \Gamma \rVert_{\alpha'}.
\ee
\end{corollary}
Let us turn now to the proof of Theorem \ref{thm1}.
We have
\be
\mu^{+}_{\Lambda,0,\beta}(\sigma_0=-1)\le \mu^{+}_{\Lambda,0,\beta}[0\in \Gamma]
=\frac{1}{Z^{+}_{\Lambda,\beta}}\sum_{\Gamma \ni 0}\sum_{\underline{\Gamma} \ni \Gamma}e^{-\beta H^+_{\alpha}[\underline{\Gamma}]}
\ee

By inequality (\ref{littin}),
\be
H^+_{\alpha}[\underline{\Gamma}]-H^+_{\alpha}[\underline{\Gamma}\backslash \{ \Gamma \}]
\ge K_c(\alpha)H_{\alpha}^+[\Gamma].
\ee

Then,
\be
\begin{aligned}\label{critical}
\mu^{+}_{\Lambda,0,\beta}(\sigma_0=-1)
&\le \sum_{\Gamma \ni 0}e^{-\beta K_c(\alpha)H_{\alpha}^+[\Gamma]}.\\
\end{aligned}
\ee
Since the Hamiltonian $H^+_{\alpha}$ is monotone in $\alpha$, choose $\alpha \in [\alpha^{*},1)$ with $\alpha\ge \alpha'$ and $\alpha' \in [0,\alpha^*)$.  Then, by Corollary \ref{coro1},
\begin{equation}\label{goodbound}
H_{\alpha}^+[\Gamma]\ge \zeta_{\alpha'}\lVert \Gamma \rVert_{\alpha'}.
\end{equation}

Therefore,
$$
\begin{aligned}
\mu^{+}_{\Lambda,0,\beta}(\sigma_0=-1) &\le
 \sum_{m\ge 1}\sum_{\substack{|\Gamma|=m \\ 0\in \Gamma}}e^{-\beta K_c(\alpha)\zeta_{\alpha'}\lVert \Gamma \rVert_{\alpha'} }\\
&\le  2\sum_{m\ge 1} m e^{-\beta K_c(\alpha) \zeta_{\alpha'} m^{\alpha'}}\\
&< \frac{1}{2},
\end{aligned}
$$
for  $\beta$ large enough. The second inequality comes from the inequality of the entropy (\ref{entropy1}).

%

\subsection{Proof of Theorem \ref{thm2}}

We first prove the theorem for all $\alpha\in (0,1)$ and external fields defined by $h_x=h_*\cdot(1+|x|)^{-\gamma}$ where  $\max\{1-\alpha,1-\alpha^*\} < \gamma \leq 1$.  
The proof will be a modified version of the Peierls argument presented in the previous section. It is enough to find an appropriate  lower bound as in \eqref{goodbound}.
Given a positive integer $L$, let us consider the external field $\bar{h}_{L}=(h_{L,x})_{x\in \mathbb{Z}}$ defined as
\begin{equation}\label{externalfield}
h_{L,x} =
\begin{cases}
0 & \text{if $|x|< L$},\\
\frac{h^{*}}{(1+|x|)^ \gamma} & \text{otherwise}.
\end{cases}   
\end{equation}

We can consider this modified field without loss of generality because any local perturbation does not change the fact that the model undergoes a phase transition. In fact, any finite-energy perturbation does not destroy the presence (or the absence) of a phase transition. This is the reason that we do not need to consider $\gamma>1$ since in that case the external field is summable, see \cite{Geo} for more details.  

Let $\Gamma \in \underline{\Gamma}$ be a  contour. Note that
\be
|\mathbbm{1}_{\{\sigma_x(\underline{\Gamma})=-1\}}-\mathbbm{1}_{\{\sigma_x(\underline{\Gamma}\setminus \Gamma)=-1\}}|\le \mathbbm{1}_{\{x\in \mathbb{Z}\cap \Gamma\}}.
\ee
Again by inequality (\ref{littin}),
$$
\begin{aligned}
H^+_{\bar{h}_L,\alpha}[\underline{\Gamma}]-H^+_{\bar{h}_L,\alpha}[\underline{\Gamma} \setminus \Gamma]&=H_{\alpha}^+[\underline{\Gamma}]-H_{\alpha}^+[\underline{\Gamma} \setminus \Gamma] +\sum_{x\in \mathbb{Z}}h_{L,x}\cdot \left(\mathbbm{1}_{\{\sigma_x(\underline{\Gamma})=-1\}}-\mathbbm{1}_{\{\sigma_x(\underline{\Gamma}\setminus \Gamma)=-1\}}\right)\\
&\ge K_c(\alpha)H_{\alpha}^+[\Gamma] - \sum_{T\in \Gamma}\sum_{x\in T\cap \mathbb{Z}} |h_{L,x}|.
\end{aligned}
$$
We can find the following upper bound for the contribution of the external fields, there is $C>0$ such that
\be
\begin{aligned}\label{externalfieldbound}
\sum_{T\in \Gamma}\sum_{x\in T\cap \mathbb{Z}} |h_{L,x}|
\le \frac{C|h_*|}{1-\gamma} L^{-p} \lVert \Gamma \rVert_{1-\gamma}.
\end{aligned}
\ee
for $p=\gamma+\alpha-1>0$.

Since the Hamiltonian $H^+_{\alpha}$ is decreasing in $\alpha$, choose $\alpha' \in (0, \alpha^{*})$ with $\alpha'\le \alpha$ such that $1-\gamma <\alpha'$. Then, by Proposition \ref{main2},
\begin{equation}\label{goodbound2}
H_{\alpha'}^+[\Gamma]\ge \zeta_{\alpha'}\lVert \Gamma \rVert_{\alpha'}.
\end{equation}

Therefore,

\be
\begin{split}
H^+_{\bar{h}_L,\alpha}[\underline{\Gamma}]-H^+_{\bar{h}_L,\alpha}[\underline{\Gamma} \setminus \Gamma]& \geq K_c(\alpha) \zeta_{\alpha'}\lVert \Gamma \rVert_{\alpha'} -\frac{C|h_*|}{1-\gamma} L^{-p} \lVert \Gamma \rVert_{\alpha'}.
\end{split}
\ee

Choosing $L$ large enough such that $K_c(\alpha) \zeta_{\alpha'} - \frac{C|h_*|}{1-\gamma} L^{-p}>0 $ concludes the proof of the lower bound.

\textbf{Critical case.} For the case when $\alpha\in (0,\alpha^{*})$ and $\gamma=1-\alpha$, we can argue in a similar way. First we take $h_*$ small enough the sum (\ref{externalfieldbound}) converges to zero and then, by the same argument above, the model undergoes a phase transition. 
The argument also holds for $\alpha=0$ and $\gamma=1$, we then have the bound 
\be
\sum_{x\in \Gamma\cap \mathbb{Z}}h_{x}\le 8 |h_* | \cdot \sum_{T\in \Gamma} \log(|T|)
\ee
and follow the same argumentation as before.

\section{Concluding remarks}\label{conc}

In this paper we provided further steps on the way to a better understanding of Dyson models. We were able to remove the artificial hypothesis of $J(1)\gg 1$, which is required  in  most  of the literature about contour methods in these models, for the whole range of the exponent of the interaction $\alpha \in [0,1)$ in the direct proof of the phase transition via contours. The result suggests that this hypothesis can also be removed in all the subsequent papers based on the contour argument from \cite{CFMP}. 

Our analysis also allows us to include the case of decaying fields of the type defined by  $h_x= h_*\cdot (1+|x|)^{-\gamma}$, previously considered in short-range models in \cite{BCCP, BC, CV, BEvE}. We proved the phase transition for any $\alpha \in [0,1)$ when $\gamma>\max\{1-\alpha, 1-\alpha^* \}$, combining our estimates with a recent approach proposed in \cite{LP}. As mentioned in the introduction, if we allow for the nearest-neighbour term to be large we can obtain a larger parameter range for $\gamma$, namely  $\gamma>\max\{1-\alpha, 1-\alpha_+ \}$ where $\alpha^* < \alpha_+$. 
The entropy bound of \cite{CFMP} can be applied to contour weights $w(\Gamma)$ involving the exponential of the ``contour norm" $|| \Gamma ||_{\alpha}$. The reason why we cannot obtain a phase transition for $\gamma > 1-\alpha$ for all $\alpha\in [0,1)$ is that when we add a field, we cannot lower-bound the contour energy by a constant times some appropriate $\alpha$-contour norm for $\alpha \geq \alpha^*$. One would need another entropy bound in order to tackle this problem. This is one possibility for further research.

 The critical case seems to be similar to the standard Ising model with decaying fields, see \cite{BCCP}. In fact, for any $0 \leq\alpha <1$ we proved the phase transition when $1- \alpha^{*}< 1-\alpha = \gamma$. For the case when $h_*$ is large and the model has a nearest-neighbour ferromagnetic interaction on $\mathbb{Z}^d$, $d\geq 2$, this question is still an open problem and there is some hope to see phase uniqueness for large values of $h_*$. The situation is completely different from the Ising model on regular non-amenable trees where there is no possibility of phase transition in the critical case for models with decaying fields, see \cite{BEvE}. 
 
We conjecture, as indicated in the Introduction,  that there is  phase uniqueness at low temperature when  $\gamma < 1- \alpha $, but we don't have a rigorous proof  so far.

In the short-range Ising model in higher dimensions, uniqueness has been proven.
In \cite{BCCP}  the uniqueness for low temperatures was proven, adapting an argument from \cite{BMPZ}, afterwards the phase uniqueness for all temperatures was obtained combining this with  a result from \cite{CV} using random cluster representations for models with fields. A direct argument proving the phase uniqueness for all $\beta >0$ for Ising models with external fields decaying slowly is still unknown. We remark, by the way, that Potts models in a homogeneous field provide an example where there is both low-temperature uniqueness and high-temperature uniqueness, with a phase transition at some intermediate temperature, thus we will need to make use of some specific Ising properties.

\section*{Acknowledgements}
We thank the referee for a number of helpful remarks. RB and EE thank Jorge Littin for providing  us with his thesis and with a preliminary version of \cite{LP} and for fruitful  discussions. RB thanks  Maria Eul\'{a}lia Vares for calling his attention to the problem of getting rid of the condition  $J(1)\gg 1$ when one uses contours for Dyson models.\\
We thank Arnaud Le Ny, Marzio Cassandro and Luiz Renato Fontes for all they taught in earlier collaborations and/or in discussions and helpful suggestions on the manuscript.\\ 
EE is supported by FAPESP grants 2014/10637-9 and 2015/14434-8. RB is supported by FAPESP Grants 2016/25053-8, 2016/08518-7 and CNPq grants 312112/2015-7 and 446658/2014-6.


\begin{thebibliography}{99}

\bibitem{ACCN} M. Aizenman, J.T. Chayes, L. Chayes, C.M. Newman.  Discontinuity of the Magnetization in one-dimensional $1/|x-y|^{2}$ Ising and Potts models. \textit {J. Stat.Phys.}  \textbf{50}, Issue 1-2:1--40, 1988.


\bibitem{ABE} N. Alon, R. Bissacot, E.O. Endo. Counting Contours on Trees. \textit{Letters in Mathematical Physics} \textbf{107},  Issue 5: 887--899, 2017.

\bibitem{BB} P. N. Balister, B. Bollob\'{a}s. Counting Regions with Bounded Surface Area. \emph{Communications in Mathematical Physics}, \textbf{273}, Issue 2: 305-315, 2007. 

\bibitem{BCCP} R. Bissacot, M. Cassandro, L. Cioletti, E. Presutti. {Phase Transitions in Ferromagnetic Ising Models with Spatially Dependent Magnetic Fields}.
\emph{Communications in Mathematical Physics}, \textbf{337}, Issue 1: 41--53, 2015.

\bibitem{BC} R. Bissacot, L. Cioletti. Phase Transition in Ferromagnetic Ising Models with Non-uniform External Magnetic Fields. \textit{Journal of Statistical Physics}, \textbf{139}, Issue 5: 769--778, 2010.

\bibitem{BEvE} R. Bissacot, E.O. Endo, A.C.D.  van Enter. Stability of the Phase Transition of Critical-Field Ising Model on Cayley Trees under Inhomogeneous External Fields.\emph{Stochastic Processes and their Applications}, \textbf{127}, Issue 12: 4126-4138, 2017.

\bibitem{BEvEKLNR} R. Bissacot, E.O. Endo, A.C.D. van Enter, B. Kimura, A. Le Ny, W. Ruszel. Dyson models under renormalization and in weak fields. \textit{ArXiv:1702.02887}, 2017.

\bibitem{BEvELN} R. Bissacot, E.O. Endo, A.C.D. van Enter, A. Le Ny. Entropic repulsion and lack of the g-measure property for Dyson models. \textit{ArXiv:1705.03156}, 2017.

\bibitem{BMPZ} A. Bovier, I. Merola, E. Presutti, M. Zahradn{\'i}k. On the Gibbs Phase Rule in the Pirogov--Sinai Regime. \emph{Journal of Statistical Physics}. \textbf{114}:1235--1267, 2004.

\bibitem{CFMP} M. Cassandro, P.A. Ferrari, I. Merola, E. Presutti. Geometry of contours and Peierls estimates in $d=1$ Ising models with long range interactions. {\em Journal of Mathematical Physics.} {\bf 46}(5):0533305, 2005. 

\bibitem{CMP} M. Cassandro, I. Merola, P. Picco. Phase Separation for the Long Range One-dimensional Ising Model. {\em Journal of Statistical Physics.} {\bf 167}:351--382, 2017.

\bibitem{CMPR} M. Cassandro, I. Merola, P. Picco, U. Rozikov. One-Dimensional Ising Models with Long Range Interactions: Cluster Expansion, Phase-Separating Point. {\em Communications in Mathematical Physics.} {\bf 327}:951--991, 2014.

\bibitem{COP1} M. Cassandro, E. Orlandi, P. Picco. Phase Transition in the 1d Random Field Ising Model with Long Range Interaction. {\em Communications in Mathematical Physics.} {\bf 288}:731--744, 2009.

\bibitem{COP2} M. Cassandro, E. Orlandi, P. Picco. Typical Gibbs Configurations for the 1d Random Field Ising Model with Long Range Interaction. {\em Communications in Mathematical Physics.} {\bf 309}:229--253, 2012.

\bibitem{CV} L. Cioletti, R. Vila.  Graphical Representations for Ising and Potts Models in General External Fields. \textit{Journal of Statistical Physics}. \textbf{162}, Issue 1: 81--122, 2016.

\bibitem{Dys1} F.J. Dyson. Existence of a Phase Transition in a One-Dimensional Ising ferromagnet. {\em Communications in Mathematical Physics.} {\em 12}:91--107, 1969.

\bibitem{Dys2} F.J. Dyson. An Ising ferromagnet with discontinuous long-range order.  {\em Communications in Mathematical Physics.} {\bf 21}:269--283, 1971.

\bibitem{Dys3} F.J. Dyson. Existence and Nature of Phase Transition in One-Dimensional Ising Ferromagnets. {\em SIAM-AMS Proceedings}. \textbf{5}:1--12, 1972.

\bibitem{FILS} J. Fr\"ohlich, R. Israel, E.H. Lieb, B. Simon. Phase transitions and reflection positivity. I. General theory and long range lattice models. {\em Communications in Mathematical Physics. } \textbf{62}:1--34, 1978.

\bibitem{frsP} J. Fr\"ohlich, T. Spencer. The phase transition in the one-dimensional Ising model with $1/r^2$ interaction energy. {\em Communications in Mathematical Physics.} {\bf 84}:87--101, 1982.

\bibitem{Geo} H.-O. Georgii. { Gibbs Measures and Phase Transitions}. de Gruyter, 1988 $\&$ 2011.

\bibitem{IN} J.Z. Imbrie, C.M. Newman. An intermediate phase with slow decay of correlations in one dimensional $1/|x-y|^2$ percolation, Ising and Potts models. {\em Communications in Mathematical Physics.} {\bf 118}:303--336, 1988.

\bibitem{Joh} K. Johansson. Condensation of a one-dimensional lattice gas. {\em Communications in Mathematical Physics.} {\bf 141}: 41--61, 1991.

\bibitem{KacT} M. Kac, C.J. Thompson. Critical Behaviour of Several Lattice Models with Long-Range Interaction. {\em J. Math. Phys.} {\bf 10}: 1373--1386, 1969.

\bibitem{Ke} A. Kerimov. The one-dimensional long-range ferromagnetic Ising model with a periodic external field. {\em Physica A.} {\bf 391}:2931--2935, 2012.

\bibitem{LY} T.D. Lee, C.N. Yang. Statistical Theory of Equations of State and Phase Transitions. II. Lattice Gas and Ising Model, {\em Phys. Review} {\bf 87}: 410--419, 1952.

\bibitem{Litthes} J. Littin. Quasi stationary distributions when infinity is an entrance boundary, optimal conditions for phase transition in 1 dimensional Ising model by Peierls argument and its consequences. Marseille, Ph.D. thesis, 2013.

\bibitem{Lit2} J. Littin: In preparation (private communication), 2017.

\bibitem{LP} J. Littin, P. Picco. Quasi-additive estimates on the Hamiltonian for the one-dimensional long-range Ising model. \emph{Journal of Mathematical Physics.} {\bf 58}, 073301, 2017.

\bibitem{Pe} R. Peierls. { On Ising's model of ferromagnetism}. \emph{Mathematical Proceedings of the Cambridge Philosophical Society.} \textbf{32}, 477--481, 1936.

\bibitem{Ra} A. Raoufi. Translation-Invariant Gibbs States of Ising model: General Setting. ArXiv:1710.07608, 2017.

\end{thebibliography}
\end{document}